\documentclass[11pt]{article}
\usepackage{tikz}
\usetikzlibrary{positioning, arrows.meta}
\usepackage[utf8]{inputenc}
\usepackage{amsmath, amssymb, amsthm}
\usepackage{graphicx}
\usepackage{xcolor}  

\newtheorem{theorem}{Theorem}[section]

\newtheorem{definition}[theorem]{Definition}
\newtheorem{example}[theorem]{Example}
\newtheorem{proposition}[theorem]{Proposition}

\begin{document}
	
	\begin{center}
		\LARGE\textbf{Foundations of information theory for coding theory}
		
		\vspace{0.5cm}
		\large\textbf{El Mahdi Mouloua ~~~~~Essaid Mohamed}  \\
		\large\textbf{Lecture note}  \\
		\vspace{0.2cm}
	\end{center}
	\begin{abstract}
Information theory is introduced in this lecture note with a particular emphasis on its relevance to algebraic coding theory. The document develops the mathematical foundations for quantifying uncertainty and information transmission by building upon Shannon's pioneering formulation of information, entropy, and channel capacity. Examples, including the binary symmetric channel, illustrate key concepts such as entropy, conditional entropy, mutual information, and the noisy channel model. Furthermore, the note describes the principles of maximum likelihood decoding and Shannon's noisy channel coding theorem, which characterizes the theoretical limits of reliable communication over noisy channels. Students and researchers seeking a connection between probabilistic frameworks of information theory and structural and algebraic techniques used in modern coding theory will find this work helpful.
	\end{abstract}
	\section{Introduction}
In 1948, Claude Shannon developed an information theory that revolutionized the field. According to his work \cite{Shannon1948}, Shannon proposed a mathematical theory of communication. It reveals that we can achieve a reliable communication through a noisy channel by determining the limits of feasible data that we can send while being almost certain of recovering the original message from the distorted signal. In this section we present a brief description to information theory. For more details on the theory of information authors can refer to \cite{resa, Shannon1948, nober, ball}. \\
Let's consider a probabilistic experiment in which a discrete random variable $X$ is observed. Let $x_{1},x_{2},\ldots,x_{N}$ the possible values that can $X$ takes, with probabilities $p_{1},p_{2},\ldots,p_{N},$ respectively. In this situation we dispose of a probability function such that : \[\mathbb{P}(X=x_{i})=p_{i}.\]
In the context of communication, the possible values of $X$ are the symbols which can be sent (or received). \\
It follows that $p_{i}$ is strictly greater than $0$, and $\sum\limits_{i=1}^{N}p_{i}=1$. The mathematical model proposed by Shannon aims to construct a mathematical measure of the information conveyed by a message.  To clarify this thing let consider the following preliminary example, suppose that $X$ is random discrete takes on the values $1, 2, 3, 4, 5$ with equal probability. We want to determine how much information the statement $1 < X < 3$ provides about the value of $X$. Initially, without any additional information, the probability of correctly guessing $X$ is $\frac{1}{6}$. After learning that $X$ can only be $1$ , $2$ or $3$, the probability of a correct guess increases. 
This shows that the uncertainty about the value of $X$ has decreased. The condition $1 < X < 3$ therefore reduces our uncertainty about the actual value of $X$. If we could formalize the notion of uncertainty, we would be able to measure precisely how much information this statement gives us. To measure information, let us consider the following function : 
	\[
	H : ]0,1] \to [0,+\infty[,
	\]
	
\noindent such that
	\begin{itemize}
		\item[(a)] $H(u)$ is a continuous decreasing function and $H(u)=0$ if $u=1$,
		\item[(b)] $H(uv)=H(u)+H(v)$.
	\end{itemize}
In the case where $H$ acts on a given probability $p$, it can be considered as a measure of information. Indeed, according to property $a$ in the construction of the function $H$, a lower probability corresponds to a larger amount of information. On the other hand, property $b$ shows that if we repeat the experiment, the information obtained from two independent outcomes is equal to the sum of the information obtained by considering each outcome separately.  
	
	\begin{proposition}
		Let $H$ be a given function with conditions $(a)$ and $(b)$. Then, one can write $H(x)$ as
		\[H(x)=-log_{r}(x),\]
		Where $r$ is a real that satisfies $r>1$. 
	\end{proposition}
\begin{proof}
 Let us consider the function $f(x)=H(e^{-x})$, by $(b)$ we have \[f(x+y)=H(e^{-(x+y)})=H(e^{-x}e^{-y})=H(e^{-x})+H(e^{-y})=f(x)+f(y)\]
 It follows that $f$ is a continuous additive function, thus one can write $f(x)=\lambda x$ for some $\lambda \in \mathbb{R}.$
 Putting $y = e^{-x}$, then $H(y)=f(-ln y)=-\lambda ln y.$  \\
 Since $ln y$ is an increasing function of $y$ and, according to $(a)$ $H$ is a decreasing function, we
 have $\lambda > 0$. \\By taking $\lambda = (ln r)^{-1},$ the proposition follows.
\end{proof}
\noindent Next we present the notion of \textbf{entropy} which quantify the measurement of information that we can get from a random variable $X$.  	
	\begin{definition}\label{def1}[Entropy]
		For a random variable $X$ with probabilities $p_i$, we define entropy as :
			\[
			H_{r}(X) = \sum\limits_{i=1}^{N}p_{i}H(p_{i})=-\sum\limits_{i=1}^{N} p_{i} \log_{r} (p_i)
			\]
		\end{definition}
\noindent In definition \ref{def1},the sum of information is weighted according to each corresponding probability.
\noindent For the case of binary information source that emits the symbols $X = 0$ and $X = 1$ with probabilities 
$\mathbb{P}\{X=0\} = p_0$ and $\mathbb{P}\{X=1\} = 1 - p_0$, 
the associated entropy is given by the \emph{Shannon function} or \emph{binary entropy}. It is written as:
\[
H_{2}(X)
= -\,p_0 \log_2(p_0)\;-\;(1-p_0)\log_2(1-p_0).
\]

	\section{Channel capacity}
	\noindent In what follows we discus the design of a channel communication based on the entropy concept. Let $S$ be a random variable which takes values from a finite set $I$ who refer to the input symbols set(Sending symbols). Let $R$ be a random variable which takes values from a finite set $O$ who refers to the output symbols set(Received symbols). Let $\{s_{1},s_{2},\ldots,s_{M}\}$ and  $\{r_{1},r_{2},\ldots,r_{N}\}$ the domains of $S$ and $R$ respectively.Let us consider the following probabilities:
	\[\phi_{i,j}=\mathbb{P}(R=r_{j}|S=s_{i}),\]
	\[\varphi_{i,j}=\mathbb{P}(S=s_{i}|R=r_{j}),\]
	\[\rho_{i,j}=\mathbb{P}(S=s_{i},R=r_{j}).\]
	  The probabilities $\phi_{i,j}$ is called \textbf{forwards probabilities}, $\varphi_{i,j}$ is called \textbf{backwards probabilities} and $\rho_{i,j}$ is called the \textbf{joint probabilities}.
	\begin{definition}[Input entropy/Output entropy]
    Let  $S$ be the input random variable with probabilities $\phi_i$ for $i \in \{1,\ldots,M\}$, and let $R$ be the output random variable with probabilities $\varphi_{j}$ $i \in \{1,\ldots,N\}.$
    \begin{itemize}
    	\item The \textbf{input entropy} is defined as \[H(S)=-\sum\limits_{i=1}^{M}\phi_{i}~\text{log}~ \phi_{i}\]
    	\item The \textbf{output entropy} is defined as \[H(R)=-\sum\limits_{i=1}^{N}\varphi_{i}~\text{log} ~\varphi_{i}\]
    \end{itemize}
\end{definition}
\noindent In the remain of this lecture we omit the symbol $r$ in the logarithm, but we suppose that it is the same for both definitions. Let now introduce the notions of input and output entropies. \\ 
In the fact that we know $R=r_{j}$ below, we investigate the conditional entropy under the fixed condition $R=r_{j}$. 
\[H(S|R=r_{j})=-\sum\limits_{i=1}^{M}P(S=s_{i}|R=r_{j})log(P(S=s_{i}|R=r_{j}))=-\sum\limits_{i=1}^{M}=\varphi_{i,j}log(\varphi_{i,j})\]

\noindent By averaging over $r_{j}$ in $R$, we introduce the conditional entropy, which is the average information about $S$ given $R$, denoted $H(R|S)$. It is given by:
\[
H(S|R) = -\sum_{i,j} r_j \, \varphi_{ij} \log \varphi_{ij}
\]
In the same way, we introduce the conditional entropy of the average information of $R$ knowing $S$ is given by:
\[
H(R|S) = -\sum_{i,j} s_j \, \phi_{ij} \log \phi_{ij}
\]
With these conditional probabilities, we define the \textbf{mutual information} is given by :
\[I(R, S) = H(R) - H(R|S) = H(S) - H(S|R)\] 
  The mutual information measures the amount of information that is transmitted across the channel from the input to the output for a given information source. The channel \textbf{capacity} is \[ C = \text{max} ~I(S,R) .\]
  According to Shannon, in the case where the input entropy $H(S)$ is smaller than the channel capacity $C$, Hence, it is possible to transmit information over a noisy channel with an arbitrarily small probability of error. \\ Therefore, the channel capacity $C$ precisely measures the maximum amount of information that can be reliably conveyed through the channel. 
\begin{example}[Binary symmetric channel]
A binary symmetric channel with bit error probability $p$ transmits the binary symbol $S = 0$ or $S = 1$ in a correct manner with probability $1-p$, whereas the incorrect binary symbol $R = 1$ or $R = 0$ is emitted with probability $p$
	\begin{center}
		\includegraphics[width=0.8\linewidth]{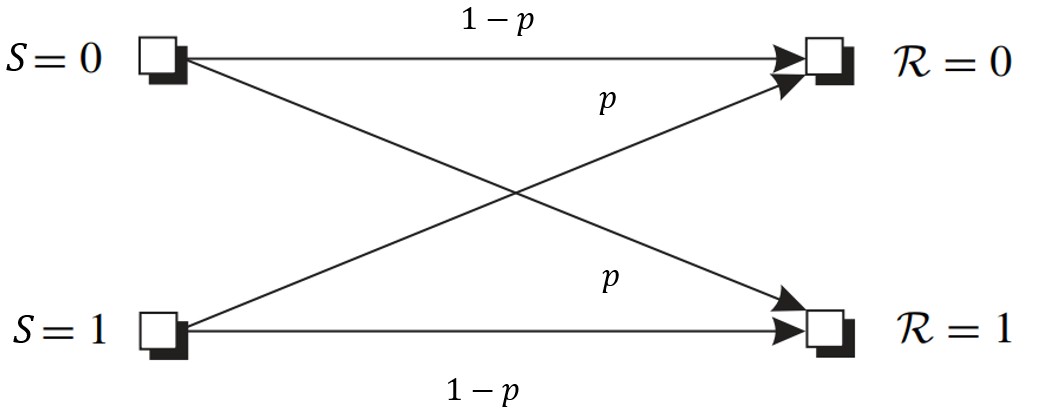}
	\end{center}

	\noindent where $p$ corresponds to the probability that a given symbol does not change. Hence, the channel capacity is 
	\[C  = 1 + p~\text{log}_{2}~p+(1-p)~\text{log}_{2}(1-p),\]\\ 
\end{example}
\noindent In order to introduce the fundamental theorem of Shannon for binary symmetric channel, we review the maximum likelihood decoding method. Let
\begin{align*}
	\mathcal{D} : R=\{r_{1},r_{2},\ldots,r_{N}\} &\to \{s_{1},s_{2},\ldots,s_{N}\}  \\
	r_{i} &\mapsto s_{\sigma(i)}
\end{align*}
 A decoding is correct is characterized by the probability 
\[\varphi_{\sigma(j),j}=\mathbb{P}(S=s_{\sigma(j)}|R=r_{j}),\]
hence the average probability $P_{\mathcal{D}}$ related to a correct decoding is 
\[P_{\mathcal{D}}=\sum\limits_{j}\varphi_{j}\varphi_{\sigma(j),j}.\]
Furthermore, the case where $\sigma$ satisfies $\phi_{\sigma(j),j} \geqslant \phi_{i,j}$  $\forall i$, we refer to the maximum likelihood decoding.
\begin{proposition}
	For a binary symmetric channel and a block code $C \subseteq \{0,1\}^n$, the maximum likelihood decoder assigns to each received vector $v$ the codeword $u \in C$ that is closest to $v$ in the Hamming metric.
\end{proposition}
\noindent We conclude this section by introducing the Shannon's noisy channel coding theorem.
\begin{theorem}
	Let $\delta$ be any small positive real number and let $R$ be a positive rate satisfying $R < C$. For all sufficiently large block lengths $n$, there exists a code of length $n$ and rate $R$ such that, under maximum likelihood decoding, the probability $P_{\mathcal{D}}$ of correct decoding exceeds $1 - \delta$.
\end{theorem}


\begin{thebibliography}{9}

	\bibitem{Shannon1948}
	C.~E. Shannon,
	\textit{A Mathematical Theory of Communication},
	Bell System Technical Journal, vol. 27, pp. 379--423, 1948.
	\bibitem{resa} Reza, F. M. (1994). An introduction to information theory. Courier Corporation.
	\bibitem{nober}Neubauer, A., Freudenberger, J., \& Kuhn, V. (2007). Coding theory: algorithms, architectures and applications. John Wiley \& Sons.
	\bibitem{ball} Ball, S. (2020). A course in algebraic error-correcting codes. Birkhäuser.
\end{thebibliography}
\end{document}